\documentclass[11pt]{article}%
\usepackage{amsfonts}
\usepackage{amsmath}
\usepackage{amssymb}
\usepackage{graphicx}%
\providecommand{\U}[1]{\protect\rule{.1in}{.1in}}
\newtheorem{theorem}{Theorem}
\newtheorem{acknowledgement}[theorem]{Acknowledgement}

\newtheorem{notation}[theorem]{Notation}

\newenvironment{proof}[1][Proof]{\noindent\textbf{#1.} }{\ \rule{0.5em}{0.5em}}
\begin{document}

\title{Short-Time Propagators and the Born--Jordan Quantization Rule}
\author{Maurice A. de Gosson\thanks{maurice.de.gosson@univie.ac.at}\\University of Vienna\\Faculty of Mathematics (NuHAG)}
\maketitle

\begin{abstract}
We have shown in previous work that the equivalence of the Heisenberg and
Schr\"{o}dinger pictures of quantum mechanics requires the use of the Born and
Jordan quantization rules. In the present work we give further evidence that
the Born--Jordan rule is the correct quantization scheme for quantum
mechanics. For this purpose we use correct short-time approximations to the
action functional, which lead to the desired quantization of the classical Hamiltonian.

\end{abstract}

\section{Motivation and Background}

\subsection{Weyl vs Born and Jordan}

There have been several attempts in the literature to find the
\textquotedblleft right\textquotedblright\ quantization rule for observables
using either algebraic or analytical techniques
\cite{Dewey,Hall,khan,kuma,mado1,mado2,Shewell}. In a recent paper \cite{FPBJ}
we have analyzed the Heisenberg and Schr\"{o}dinger pictures of quantum
mechanics, and shown that if one postulates that both theories are equivalent,
then one must use the Born--Jordan quantization rule
\begin{equation}
\text{(BJ) \ \ \ }x^{m}p^{\ell}\longrightarrow\frac{1}{m+1}\sum_{k=0}%
^{m}\widehat{x}^{k}\widehat{p}^{\ell}\widehat{x}^{m-k}, \label{bj1}%
\end{equation}
and \emph{not} the Weyl rule\footnote{To be accurate, it was McCoy
\cite{mccoy} who showed that Weyl's quantization scheme leads to formula
(\ref{w2}).}%
\begin{equation}
\text{(Weyl) \ \ \ \ }x^{m}p^{\ell}\longrightarrow\frac{1}{2^{m}}\sum
_{k=0}^{m}\binom{m}{k}\widehat{x}^{k}\widehat{p}^{\ell}\widehat{x}^{m-k}
\label{w2}%
\end{equation}
for monomial observables. The Born--Jordan and Weyl rules yield the same
result only if $m<2$ or $\ell<2$; for instance in both cases the quantization
of the product $xp$ is $\frac{1}{2}(\widehat{x}\widehat{p}+\widehat{p}%
\widehat{x})$. One can also show that the product $pf(x)$ is, for any smooth
function $f$ of position alone, given in both cases by the symmetric rule
\[
pf(x)\longrightarrow\frac{1}{2}(\widehat{p}f(x)+f(x)\widehat{p}).
\]
It follows that if $H$ is a Hamiltonian of the type
\[
H=\sum_{j=1}^{n}\frac{1}{2m_{j}}(p_{j}-A_{j}(x))^{2}+V(x)
\]
one can use either the Weyl or the Born--Jordan prescriptions to get the the
corresponding quantum operator, which yields the familiar expression%
\[
\widehat{H}=\sum_{j=1}^{n}\frac{1}{2m_{j}}\left(  -i\hbar\frac{\partial
}{\partial x_{j}}-A(x)\right)  ^{2}+V(x).
\]
(See Section \ref{sec3}). Since this Hamiltonian is without doubt the one
which most often occurs in quantum mechanics one could ask why one should
bother about which is the \textquotedblleft correct\textquotedblright%
\ quantization. It turns out that this question is just a little bit more than
academic: there are simple physical observables which yield different
quantizations in the Weyl and Born--Jordan schemes. One interesting example is
that of the squared angular momentum: writing $\mathbf{r}=(x,y,z)$ and
$\mathbf{p}=(p_{x},p_{y},p_{z})$ the square of the classical angular momentum
\begin{equation}
\ell=(yp_{z}-zp_{y})\mathbf{i}+(zp_{x}-xp_{z})\mathbf{j}+(xp_{y}%
-yp_{x})\mathbf{k} \label{elle}%
\end{equation}
is the function $\ell^{2}=\ell_{x}^{2}+\ell_{y}^{2}+\ell_{z}^{2}$ where
\begin{equation}
\ell_{x}^{2}=x^{2}p_{y}^{2}+y^{2}p_{x}^{2}-2xp_{x}yp_{y}%
\end{equation}
and so on. The Weyl quantization of $\ell_{x}^{2}$ is
\begin{equation}
(\widehat{\ell_{x}^{2}})_{\mathrm{W}}=\widehat{x}^{2}\widehat{p}_{y}%
^{2}+\widehat{x}_{y}^{2}\widehat{p}_{x}^{2}-\tfrac{1}{2}(\widehat{x}%
\widehat{p}_{x}+\widehat{p}_{x}\widehat{x})(\widehat{y}\widehat{p}%
_{y}+\widehat{p}_{y}\widehat{y}) \label{lw}%
\end{equation}
while its Born--Jordan quantization is%
\begin{equation}
(\widehat{\ell_{x}^{2}})_{\mathrm{BJ}}=\widehat{x}^{2}\widehat{p}_{y}%
^{2}+\widehat{x}_{y}^{2}\widehat{p}_{x}^{2}-\tfrac{1}{2}(\widehat{x}%
\widehat{p}_{x}+\widehat{p}_{x}\widehat{x})(\widehat{y}\widehat{p}%
_{y}+\widehat{p}_{y}\widehat{y})-\tfrac{1}{6}\hbar^{2};
\end{equation}
similar relations are obtained for $\ell_{y}^{2}$ and $\ell_{z}^{2}$ so that,
in the end,
\begin{equation}
(\widehat{\ell^{2}})_{\mathrm{W}}-(\widehat{\ell^{2}})_{\mathrm{BJ}}=\tfrac
{1}{2}\hbar^{2}.
\end{equation}
This discrepancy has been dubbed the \textquotedblleft angular momentum
dilemma\textquotedblright\cite{dasp}; in \cite{dilemma} we have discussed this
apparent paradox and shown that it disappears if one systematically uses
Born--Jordan quantization.

\subsection{The Kerner and Sutcliffe approach to quantization}

As we have proven in \cite{FPBJ,SPRINGER}, Heisenberg's matrix mechanics
\cite{hei}, as rigorously constructed by Born and Jordan in \cite{bj} and
Born, Jordan, and Heisenberg in \cite{bjh}, explicitly requires the use of the
quantization rule (\ref{bj1}) to be mathematically consistent, a fact which
apparently has escaped the attention of physicists, and philosophers
or\ historians of Science. In the present paper, we will show that the Feynman
path integral approach is another genuinely physical motivation for
Born--Jordan quantization of arbitrary observables; it corrects previous
unsuccessful attempts involving path integral arguments which \textit{do not
work} for a reason that will be explained. One of the most convincing of these
attempts is the paper \cite{Kerner} by Kerner and Sutcliffe. Elaborating on
previous work of Garrod \cite{Garrod} Kerner and Sutcliffe tried to justify
the Born--Jordan rule as the unique possible quantization (see Steven
Kauffmann's \cite{kau,Kauffmann} brilliant discussion of this work). Assuming
that $\widehat{H}$ is the quantization of some general Hamiltonian $H$, they
write as is usual in the theory of the phase space Feynman integral the
propagator as%
\begin{equation}
\langle x|e^{-\frac{i}{\hbar}\widehat{H}t}|x^{\prime}\rangle=\lim
_{N\rightarrow\infty}\int dx_{N-1}\cdot\cdot\cdot dx_{1}%
{\textstyle\prod\nolimits_{k=1}^{N}}
\langle x_{k}|e^{-\frac{i}{\hbar}\widehat{H}\Delta t}|x_{k-1}\rangle\label{1}%
\end{equation}
where $x_{N}=x$ and $x_{0}=x^{\prime}$ are fixed and $\Delta t=t/N$. They
thereafter use the approximation%
\begin{equation}
\langle x_{k}|e^{-\frac{i}{\hbar}\widehat{H}\Delta t}|x_{k-1}\rangle
\thickapprox\frac{1}{2\pi\hbar}\int e^{\frac{i}{\hbar}\overline{S}%
(x,x^{\prime},p,\Delta t)}dp \label{2}%
\end{equation}
the function $\overline{S}$ being given by
\begin{equation}
\overline{S}(x,x^{\prime},p,\Delta t)=p(x-x^{\prime})-\overline{H}%
(x,x^{\prime},p)\Delta t \label{action}%
\end{equation}
where $\overline{H}$ is the time average of $H$ over $p$ fixed and $x=x(t)$,
that is%
\begin{equation}
\overline{H}(x,x^{\prime},p)=\frac{1}{\Delta t}\int_{0}^{\Delta t}H(x^{\prime
}+s\frac{x-x^{\prime}}{\Delta t},p)ds. \label{3}%
\end{equation}
Notice that introducing the dimensionless parameter $\tau=s/\Delta t$, formula
(\ref{3}) can be written in the more convenient form
\begin{equation}
\overline{H}(x,x^{\prime},p)=\int_{0}^{1}H(\tau x+(1-\tau)x^{\prime},p)ds
\label{Hbar}%
\end{equation}
which is the usual mathematical definition of Born--Jordan quantization: see
de Gosson \cite{58,SPRINGER} and de Gosson and Luef \cite{golu1}.

Taking the limit $\Delta t\rightarrow0$ the operator $\widehat{H}$ can then be
explicitly and uniquely determined, and Kerner and Sutcliffe show that in
particular this leads to the Born--Jordan ordering (\ref{bj1}) when their
Hamiltonian $H$ is a monomial $x^{m}p^{\ell}$. Unfortunately (as
immediately\footnote{Cohen's rebukal was published in the same volume of
\textit{J. Math. Phys.} in which Kerner and Sutcliffe published their
results.} noted by Cohen \cite{Cohenwrong}) there are many a priori equally
good constructions of the Feynman integral, leading to other quantization
rules. In fact, argues Cohen, there is a great freedom of choice in
calculating the action $p(x-x^{\prime})-\overline{H}$ appearing in the
right-hand side of (\ref{3}). For instance, one can choose
\begin{equation}
S(x,x^{\prime},p,\Delta t)=p(x-x^{\prime})-H(\tfrac{1}{2}(x+x^{\prime
}),p)\Delta t
\end{equation}
which leads for $x^{m}p^{\ell}$ to Weyl's rule (\ref{w2}), or one can choose%
\begin{equation}
S(x,x^{\prime},p,\Delta t)=p(x-x^{\prime})-\tfrac{1}{2}(H(x,p)+H(x^{\prime
},p))\Delta t,
\end{equation}
which leads to the symmetric rule%
\begin{equation}
x^{m}p^{\ell}\longrightarrow\frac{1}{2}(\widehat{x}^{m}\widehat{p}^{\ell
}+\widehat{p}^{\ell}\widehat{x}^{m}).
\end{equation}
This ambiguity shows -- in an obvious way -- that Feynman path integral theory
does not lead to an uniquely defined quantization scheme for observables.
However -- and this is the main point of the present paper -- while Cohen's
remark was mathematically justified, Kerner and Sutcliffe's insight was right
(albeit for the wrong reason).

\subsection{What we will do}

It turns out that the formula (\ref{action}) for the approximate action that
Kerner and Sutcliffe \textquotedblleft guessed\textquotedblright\ has been
justified independently (in another context) by Makri and Miller
\cite{makmil1,makmil2} and the present author \cite{ICP} by rigorous
mathematical methods. This formula is actually the \emph{correct}
approximation to action up to order $O(\Delta t^{2})$ (as opposed to the
\textquotedblleft midpoint rules\textquotedblright\ commonly used in the
theory of the Feynman integral which yield much cruder approximations); it
follows that Kerner and Sutcliffe's formula (\ref{2}) indeed yields a correct
approximation of the infinitesimal propagator $\langle x_{k}|e^{-\frac
{i}{\hbar}\widehat{H}\Delta t}|x_{k-1}\rangle$, in fact the \emph{best one
}for calculational purposes since it ensures a swift convergence of numerical
schemes. This is because for short times $\Delta t$ the solution of
Schr\"{o}dinger's equation
\begin{equation}
i\hbar\frac{\partial\psi}{\partial t}(x,t)=\left[  \sum_{j=1}^{n}\frac
{-\hbar^{2}}{2m_{j}}\frac{\partial^{2}}{\partial x_{j}^{2}}+V(x)\right]
\psi(x,t) \label{a}%
\end{equation}
with initial condition $\psi(x,0)=\psi_{0}(x)$ is given by the asymptotic
formula
\begin{equation}
\psi(x,\Delta t)=\int\overline{K}(x,x^{\prime},\Delta t)\psi_{0}(x^{\prime
})d^{n}x^{\prime}+O(\Delta t^{2}); \label{b}%
\end{equation}
the approximate propagator $\overline{K}$ being defined, for arbitrary time
$t$, by%
\begin{equation}
\overline{K}(x,x^{\prime},t)=\left(  \tfrac{1}{2\pi\hbar}\right)  ^{n}\int
\exp\left(  \tfrac{i}{\hbar}\left[  p(x-x^{\prime})-(H_{\mathrm{free}%
}(p)+\overline{V}(x,x^{\prime}))t\right]  \right)  d^{n}p, \label{c}%
\end{equation}
where, by definition, $H_{\mathrm{free}}(p)$ is the free particle Hamiltonian
function, and the two-point function
\[
\overline{V}(x,x^{\prime})=\int_{0}^{1}V(\tau x+(1-\tau)x^{\prime})d\tau
\]
is the average value of the potential $V$ on the line segment $[x^{\prime},x]$.

\begin{itemize}
\item In Section \ref{sec1} we discuss the accuracy of Kerner and Sutcliffe's
propagator by comparing it with the more familiar Van Vleck propagator; we
show that for small times both are approximations to order $O(t^{2})$ to the
exact propagator of Schr\"{o}dinger's equation.

\item In Section \ref{sec2} we show that if one assume's that short-time
evolution of the wavefunction (for an arbitrary Hamiltonian $H$) is given by
the Kerner and Sutcliffe propagator, then $H$ must be quantized following the
rule (\ref{Hbar}); we thereafter show that when $H$ is a monomial
$x^{m}p^{\ell}$ then the corresponding operator is given by the Born--Jordan
rule (\ref{bj1}), \emph{not} by the Weyl rule \ref{w2}.
\end{itemize}

\begin{notation}
The generalized position and momentum vectors are $x=(x_{1},...,x_{n})$ and
$p=(p_{1},...,p_{n})$; we set $px=p_{1}x_{1}+\cdot\cdot\cdot+p_{n}x_{n}$. We
denote by $\widehat{x}_{j}$ the operator of multiplication by $x_{j}$ and by
$\widehat{p}_{j}$ the momentum operator $-i\hbar(\partial/\partial x_{j})$.
\end{notation}

\section{On Short-Time Propagators\label{sec1}}

In this section we only consider Hamiltonian functions of the type
\textquotedblleft kinetic energy plus potential\textquotedblright:
\begin{equation}
H(x,p)=H_{\mathrm{free}}(p)+V(x)\text{ \ , \ }H_{\mathrm{free}}(p)=\sum
_{j=1}^{n}\frac{1}{2m_{j}}p_{j}^{2}. \label{Hfree}%
\end{equation}

\subsection{The Van Vleck Propagator}

Consider a Hamiltonian function of the type (\ref{Hfree}) above; the
corresponding Schr\"{o}dinger equation is%
\begin{equation}
i\hbar\frac{\partial\psi}{\partial t}(x,t)=\left[  \sum_{j=1}^{n}\frac
{-\hbar^{2}}{2m_{j}}\frac{\partial^{2}}{\partial x_{j}^{2}}+V(x)\right]
\psi(x,t). \label{schrodeq}%
\end{equation}
We will denote by $K(x,x^{\prime},t)=\langle x|e^{-\frac{i}{\hbar}\widehat
{H}t}|x^{\prime}\rangle$ the corresponding exact propagator:
\begin{equation}
\psi(x,t)=\int K(x,x^{\prime},t)\psi_{0}(x^{\prime})d^{n}x^{\prime}
\label{kernel}%
\end{equation}
where$\ $with $\psi_{0}(x)$ is the value of $\psi$ at time $t=0$. The function
$K(x,x^{\prime},t)$ must thus satisfy the boundary condition
\begin{equation}
\lim_{t\rightarrow0}K(x,x^{\prime},t)=\delta(x-x^{\prime}). \label{cond1}%
\end{equation}

It is well-known (see \textit{e.g}. Gutzwiller \cite{Gutz}, Schulman
\cite{schulman}, de Gosson \cite{ICP}, Maslov and Fedoriuk \cite{MF}) that for
short times an approximate propagator is given by Van Vleck's formula%
\begin{equation}
\widetilde{K}(x,x^{\prime},t)=\left(  \tfrac{1}{2\pi i\hbar}\right)
^{n/2}\sqrt{\rho(x,x^{\prime},t)}e^{\frac{i}{\hbar}S(x,x^{\prime},t)}
\label{vanvleck}%
\end{equation}
where
\begin{equation}
S(x,x^{\prime},t)=\int_{0}^{t}\left(
{\textstyle\sum_{j=1}^{n}}
\tfrac{1}{2}m_{j}\dot{x}_{j}(s)^{2}-V(x(s)\right)  ds \label{actionex}%
\end{equation}
is the action along the classical trajectory leading from $x^{\prime}$ at time
$t^{\prime}=0$ to $x$ at time $t$ (there is no sum over different classical
trajectories because only one trajectory contributes in the limit
$t\rightarrow0$ \cite{makmil1}) and%
\begin{equation}
\rho(x,x^{\prime},t)=\det\left(  -\frac{\partial^{2}S(x,x^{\prime}%
,t)}{\partial x_{j}\partial x_{jk}^{\prime}}\right)  _{1\leq j,k\leq n}
\label{rho}%
\end{equation}
is the Van Vleck density of trajectories \cite{ICP,Gutz,schulman}; the
argument of the square root is chosen so that the initial condition
(\ref{cond1}) is satisfied \cite{ICP,Birkbis}. It should be emphasized that
although the Van Vleck propagator is frequently used in semiclassical
mechanics, it has nothing \textquotedblleft semiclassical\textquotedblright%
\ \textit{per se}, since it is genuinely an approximation to the exact
propagator for small $t$ -- not just in the limit $\hbar\rightarrow0$. In fact:

\begin{theorem}
\label{Thm2}Let $\widetilde{\psi}$ be given by%
\[
\widetilde{\psi}(x,t)=\int\widetilde{K}(x,x^{\prime},t)\psi_{0}(x^{\prime
})d^{n}x^{\prime}%
\]
where $\psi_{0}$ is a tempered distribution. Let $\psi$ be the exact solution
of Schr\"{o}dinger's equation with initial datum $\psi_{0}$. We have%
\begin{equation}
\psi(x,t)-\widetilde{\psi}(x,t)=O(t^{2}). \label{estpsitilde}%
\end{equation}
In particular, the Van Vleck propagator $\widetilde{K}(x,x^{\prime},t)$ is an
$O(t^{2})$ approximation to the exact propagator $K(x,x^{\prime},t)$:%
\begin{equation}
K(x,x^{\prime},t)-\widetilde{K}(x,x^{\prime},t)=O(t^{2}) \label{kk}%
\end{equation}
for $t\rightarrow0$ and hence
\[
\lim_{t\rightarrow0}\widetilde{K}(x,x^{\prime},t)=\delta(x-x^{\prime}).
\]

\end{theorem}

\begin{proof}
Referring to de Gosson \cite{ICP} (Lemma 241) for details, we sketch the main
lines in the case $n=1$. Assuming that $\psi_{0}$ belongs to the Schwartz
space $\mathcal{S}(\mathbb{R}^{n})$ of rapidly decreasing functions, one
expands the solution $\psi$ of Schr\"{o}dinger's equation to second order:
\[
\psi(x,t)=\psi_{0}(x)+\frac{\partial\psi}{\partial t}(x,0)t+O(t^{2}).
\]
Taking into account the fact that $\psi$ is a solution of Schr\"{o}dinger's
equation this can be rewritten
\begin{equation}
\psi(x,t)=\left[  1+\frac{t}{i\hbar}\left(  -\frac{\hbar^{2}}{2m}%
\frac{\partial^{2}}{\partial x^{2}}+V(x)\right)  \right]  \psi_{0}%
(x)+O(t^{2}). \label{icp1}%
\end{equation}
Expanding the exponential $e^{iS/\hbar}$ in Van Vleck's formula
(\ref{vanvleck}) at $t=0$ one shows, using the estimate (\ref{estimateICP}) in
Theorem \ref{Thm1}, that we also have%
\begin{equation}
\widetilde{\psi}(x,t)=\left[  1+\frac{t}{i\hbar}\left(  -\frac{\hbar^{2}}%
{2m}\frac{\partial^{2}}{\partial x^{2}}+V(x)\right)  \right]  \psi
_{0}(x)+O(t^{2}); \label{icp2}%
\end{equation}
comparison with (\ref{icp1}) implies that $\psi(x,t)-\widetilde{\psi
}(x,t)=O(t^{2})$. By density of the Schwartz space in the class of tempered
distributions $\mathcal{S}^{\prime}(\mathbb{R}^{n})$ the estimate
(\ref{estpsitilde}) is valid if one chooses $\psi_{0}(x)=\delta(x-x_{0})$,
which yields formula (\ref{kk}) since we have%
\[
\int\widetilde{K}(x,x^{\prime},t)\delta(x-x_{0})d^{n}x^{\prime}=\widetilde
{K}(x,x_{0},t)
\]
and%
\[
\int K(x,x^{\prime},t)\delta(x-x_{0})d^{n}x^{\prime}=K(x,x_{0},t).
\]

\end{proof}

Let us briefly return to the path integral. Replacing the terms $\langle
x_{k}|e^{-\frac{i}{\hbar}\widehat{H}\Delta t}|x_{k-1}\rangle$ in the product
formula (\ref{1}) with $\widetilde{K}(x_{k-1},x_{k-1},\Delta t)$ one shows,
using the Lie--Trotter formula \cite{ICP,schulman}, that the exact propagator
$K(x,x^{\prime},t)=\langle x|e^{-\frac{i}{\hbar}\widehat{H}t}|x^{\prime
}\rangle$ is given by
\begin{equation}
\langle x|e^{-\frac{i}{\hbar}\widehat{H}t}|x^{\prime}\rangle=\lim
_{N\rightarrow\infty}\int dx_{N-1}\cdot\cdot\cdot dx_{1}%
{\textstyle\prod\nolimits_{k=1}^{N}}
\widetilde{K}(x_{k-1},x_{k-1},\Delta t).
\end{equation}
This formula is often taken as the starting point of path integral arguments:
observing that the expression (\ref{vanvleck}) is in most cases\footnote{The
free particle and the harmonic oscillator are noticeable cases where the
action integral can be explicitl claculate and thus yields an exlicit formula
for the propagator, but mathematically speaking this fact is rather a
consequence of the theory of the metaplectic group \cite{ICP,Birkbis}}
difficult to calculate (it implies the computation of an action integral,
which can be quite cumbersome) people working in the theory of the Feynman
integral replace the exact action $S(x,x^{\prime},t)$ in (\ref{vanvleck}) with
approximate expressions, for instance the \textquotedblleft midpoint
rules\textquotedblright\ that we will be discussed below.\ Now, one should be
aware that this \textit{legerdemain} works\textit{, }because when taking the
limit $N\rightarrow\infty$ one indeed obtains the correct propagator, but it
does \emph{not} imply that these midpoint rules are accurate approximations to
$S(x,x^{\prime},t)$.

\subsection{The Kerner--Sutcliffe propagator}

We showed above that the Van Vleck propagator is an approximation to order
$O(t^{2})$ to the exact propagator. We now show that the propagator proposed
by Kerner and Sutcliffe in \cite{Kerner} approximates the Van Vleck propagator
also at order $O(t^{2})$. We begin by giving a correct short-time
approximation to the action.

\begin{theorem}
\label{Thm1}The function $\overline{S}$ defined by
\begin{equation}
\overline{S}(x,x^{\prime},t)=\sum_{j=1}^{n}m_{j}\frac{(x_{j}-x_{j}^{\prime
})^{2}}{2t}-\overline{V}(x,x^{\prime})t \label{sbar}%
\end{equation}
where $\overline{V}(x,x^{\prime})$ is the average of the potential $V$ along
the line segment $[x^{\prime},x]:$
\[
\overline{V}(x,x^{\prime})=\int_{0}^{1}V(\tau x+(1-\tau)x^{\prime})d\tau.
\]
satisfies for $t\rightarrow0$ the estimate
\begin{equation}
S(x,x^{\prime},t)-\overline{S}(x,x^{\prime},t)=O(t^{2}). \label{estimateICP}%
\end{equation}

\end{theorem}

For detailed proofs we refer to the aforementioned papers
\cite{makmil1,makmil2} by Makri and Miller, and to our book \cite{ICP}; also
see de Gosson and Hiley \cite{gohi2,gohi3}. The underlying idea is quite
simple (and already appears in germ in Park's book \cite{Park}, p.438): one
remarks that the function $S=S(x,x^{\prime},t)$ satisfies the Hamilton--Jacobi
equation%
\begin{equation}
\frac{\partial S}{\partial t}+\sum_{j=1}^{n}\frac{1}{2m_{j}}\left(
\frac{\partial S}{\partial x_{j}}\right)  ^{2}+V(x)=0 \label{hamjac}%
\end{equation}
and one thereafter looks for an asymptotic solution%
\[
S(x,x^{\prime},t)=\frac{1}{t}S_{0}(x,x^{\prime})+S_{1}(x,x^{\prime}%
)t+S_{2}(x,x^{\prime})t^{2}+\cdot\cdot\cdot.
\]
Insertion in (\ref{hamjac}) then leads to
\[
S_{0}(x,x^{\prime})=\sum_{j=1}^{n}m_{j}\frac{(x_{j}-x_{j}^{\prime})^{2}}{2}%
\]
and $S_{1}(x,x^{\prime})=-\overline{V}(x,x^{\prime})$ hence (\ref{sbar}).
Notice that this procedure actually allows one to find approximations to $S$
to an arbitrary order of accuracy by solving successively the equations
satisfied by $S_{2}$ ,$S_{3},..$ (see \cite{makmil1,makmil2} for explicit formulas).

Let us now set
\[
\overline{H}(x,x^{\prime},t)=H_{\mathrm{free}}(p)+\overline{V}(x,x^{\prime})
\]
where
\[
\overline{V}(x,x^{\prime})=\int_{0}^{1}V(\tau x+(1-\tau)x^{\prime})d\tau
\]
is the averaged potential.

Let us now show that the propagator postulated by Garrod \cite{Garrod} and
Kerner and Sutcliffe \cite{Kerner} is as good an approximation to the exact
propagator as Van Vleck's is. We recall the textbook Fourier formula%
\begin{equation}
\left(  \tfrac{1}{2\pi\hbar}\right)  ^{n}\int e^{\frac{i}{\hbar}p(x-x^{\prime
})}p_{j}^{\ell}d^{n}p=\left(  -i\hbar\tfrac{\partial}{\partial x_{j}}\right)
^{\ell}\delta(x-x^{\prime}). \label{fou2}%
\end{equation}

\begin{theorem}
\label{Thm1bis}Let $\overline{K}=\overline{K}(x,x^{\prime},t)$ be defined (in
the distributional sense) by
\begin{equation}
\overline{K}(x,x^{\prime},t)=\left(  \tfrac{1}{2\pi\hbar}\right)  ^{n}\int
e^{\tfrac{i}{\hbar}(p(x-x^{\prime})-\overline{H}(x,x^{\prime},p)t)}d^{n}p.
\label{KSpropagator}%
\end{equation}
and set
\begin{equation}
\overline{\psi}(x,t)=\int\overline{K}(x,x^{\prime},t)\psi_{0}(x^{\prime}%
)d^{n}x^{\prime}. \label{psibardef}%
\end{equation}
Let $\psi$ be the solution of Schr\"{o}dinger's equation with initial
condition $\psi_{0}$. We have
\begin{equation}
\overline{\psi}(x,t)-\psi(x,t)=O(t^{2}). \label{psipsibar}%
\end{equation}
The function $\overline{K}$ is an $O(t^{2})$ approximation to the exact
propagator $K$:
\begin{equation}
K(x,x^{\prime},t)-\overline{K}(x,x^{\prime},t)=O(t^{2}). \label{essentiel}%
\end{equation}

\end{theorem}

\begin{proof}
It is sufficient to prove (\ref{psipsibar}); formula (\ref{essentiel}) follows
by the same argument as in the proof of Theorem \ref{Thm2}. To simplify
notation we assume again $n=1$; the general case is a straightforward
extension. Expanding for small $t$ the exponential in the integrand of
(\ref{KSpropagator}) we have%
\begin{align*}
\overline{K}(x,x^{\prime},t)  &  =\left(  \tfrac{1}{2\pi\hbar}\right)
^{n}\int e^{\tfrac{i}{\hbar}p(x-x^{\prime})}(1-\frac{i}{\hbar}\overline
{H}(x,x^{\prime},p)t)dp+O(t^{2})\\
&  =\delta(x-x^{\prime})-\frac{it}{\hbar}\int e^{\tfrac{i}{\hbar}%
p(x-x^{\prime})}\overline{H}(x,x^{\prime},p)dp+O(t^{2})
\end{align*}
and hence
\[
\overline{\psi}(x,t)=\psi_{0}(x)-\frac{it}{\hbar}\int e^{\tfrac{i}{\hbar
}p(x-x^{\prime})}\overline{H}(x,x^{\prime},p)dp+O(t^{2}).
\]
We have
\[
\int e^{\tfrac{i}{\hbar}p(x-x^{\prime})}\overline{H}(x,x^{\prime}%
,p)d^{n}p=\int e^{\tfrac{i}{\hbar}p(x-x^{\prime})}\left(  \frac{p^{2}}%
{2m}+\overline{V}(x,x^{\prime})\right)  dp;
\]
Using the Fourier formula (\ref{fou2}) we get
\[
\left(  \tfrac{1}{2\pi\hbar}\right)  ^{n}\int e^{\tfrac{i}{\hbar}%
p(x-x^{\prime})}\frac{p^{2}}{2m}dp=-\frac{\hbar^{2}}{2m}\frac{\partial^{2}%
}{\partial x^{2}}\delta(x-x^{\prime})
\]
and, noting that $\overline{V}(x,x)=V(x)$,
\begin{align*}
\left(  \tfrac{1}{2\pi\hbar}\right)  ^{n}\int e^{\tfrac{i}{\hbar}%
p(x-x^{\prime})}\overline{V}(x,x^{\prime})dp  &  =\overline{V}(x,x^{\prime
})\delta(x-x^{\prime})\\
&  =V(x)\delta(x-x^{\prime}).
\end{align*}
Summarizing,%
\begin{equation}
\overline{K}(x,x^{\prime},t)=\delta(x-x^{\prime})+\frac{it}{\hbar}\left(
-\frac{\hbar^{2}}{2m}\frac{\partial^{2}}{\partial x^{2}}+V(x)\right)
\delta(x-x^{\prime})+O(t^{2}) \label{estkbar}%
\end{equation}
and hence%
\[
\overline{\psi}(x,t)=\psi_{0}(x)-\frac{it}{\hbar}\left(  -\frac{\hbar^{2}}%
{2m}\frac{\partial^{2}}{\partial x^{2}}+V(x)\right)  +O(t^{2}).
\]
Comparing this expression with (\ref{icp1}) yields (\ref{essentiel}).
\end{proof}

\subsection{Comparison of short-time propagators}

We have seen above that both the Van Vleck and the Kerner--Sutcliffe
propagators are accurate to order $O(t^{2})$:%
\begin{align}
K(x,x^{\prime},t)-\widetilde{K}(x,x^{\prime},t)  &  =O(t^{2}).\\
K(x,x^{\prime},t)-\overline{K}(x,x^{\prime},t)  &  =O(t^{2})
\end{align}
and hence, of course,
\begin{equation}
\widetilde{K}(x,x^{\prime},t)-\overline{K}(x,x^{\prime},t)=O(t^{2}).
\end{equation}
Let us now study the case of the most commonly approximations to the action
used in the theory of the Feynman integral, namely the mid-point rules%
\begin{equation}
S_{1}(x,x^{\prime},t,t^{\prime})=\sum_{j=1}^{n}m_{j}\frac{(x_{j}-x_{j}%
^{\prime})^{2}}{2t}-\frac{1}{2}(V(x)+V(x^{\prime}))t \label{s1}%
\end{equation}
and
\begin{equation}
S_{2}(x,x^{\prime},t)=\sum_{j=1}^{n}m_{j}\frac{(x_{j}-x_{j}^{\prime})^{2}}%
{2t}-V(\tfrac{1}{2}(x+x^{\prime}))\Delta t. \label{s2}%
\end{equation}
We begin with a simple example, that of the harmonic oscillator
\[
H(x,p)=\frac{p^{2}}{2m}+\frac{1}{2}m^{2}\omega^{2}x^{2}%
\]
(we are assuming $n=1$). The exact value of the action is given by the
generating function%
\begin{equation}
S(x,x^{\prime},t)=\frac{m}{2\sin\omega t}((x^{2}+x^{\prime2})\cos\omega
t-2xx^{\prime}); \label{generating}%
\end{equation}
expanding the terms $\sin\omega t$ and $\cos\omega t$ in Taylor series for
$t\rightarrow0$ yields the approximation
\begin{equation}
S(x,x^{\prime},t)=m\frac{(x-x^{\prime})^{2}}{2t}-\frac{m\omega^{2}}{6}%
(x^{2}+xx^{\prime}+x^{\prime2})t+O(t^{2}). \label{wharmo1}%
\end{equation}
It is easy to verify, averaging $\frac{1}{2}m^{2}\omega^{2}x^{2}$ over
$[x^{\prime},x]$ that
\[
\overline{S}(x,x^{\prime},t)=m\frac{(x-x^{\prime})^{2}}{2t}-\frac{m\omega^{2}%
}{6}(x^{2}+xx^{\prime}+x^{\prime2})t
\]
is precisely the approximate action provided by (\ref{sbar}). If we now
instead apply the midpoint rule\ (\ref{s1}) we get%
\[
S_{1}(x,x^{\prime},t)=m\frac{(x-x^{\prime})^{2}}{2t}-\frac{m^{2}\omega^{2}}%
{4}(x^{2}+x^{\prime2})t
\]
which differs from the correct value (\ref{wharmo1}) by a term $O(\Delta t)$.
Similarly, the rule\ (\ref{s2}) yields
\[
S_{2}(x,x^{\prime},t)=m\frac{(x-x^{\prime})^{2}}{2t}-\frac{m^{2}\omega^{2}}%
{8}(x+x^{\prime})^{2}t
\]
which again differs from the correct value (\ref{generating}) by a term
$O(t)$. It is easy to understand why it is so by examining the case of a
general potential function, and to compare $\overline{V}(x,x^{\prime})$,
$\frac{1}{2}(V(x)+V(x^{\prime}))$, and $V(\tfrac{1}{2}(x+x^{\prime})$.
Consider for instance $\overline{V}(x,x^{\prime})-V(\tfrac{1}{2}(x+x^{\prime
})$. Expanding $V(x)$ in a Taylor series at $\overline{x}=\frac{1}%
{2}(x+x^{\prime})$ we get after some easy calculations%
\begin{align*}
\overline{V}(x,x^{\prime})  &  =V(\overline{x})+V^{\prime}(\overline
{x})(x-x^{\prime})+\frac{1}{2}V^{\prime\prime}(\overline{x})(x-x^{\prime}%
)^{2}+O((x-x^{\prime})^{3})\\
&  =V(\tfrac{1}{2}(x+x^{\prime})-\tfrac{1}{12}V^{\prime\prime}(\tfrac{1}%
{2}(x+x^{\prime}))(x-x^{\prime})^{3}+O((x-x^{\prime})^{3})
\end{align*}
hence $\overline{V}(x,x^{\prime})-V(\tfrac{1}{2}(x+x^{\prime})$ is different
from zero unless $x=x^{\prime}$ (or if $V(x)$ is linear) and hence the
difference between $\overline{S}(x,x^{\prime},t)$ and $S_{2}(x,x^{\prime},t)$
will always generate a term containing $t$ so that $\overline{S}(x,x^{\prime
},t)-S_{2}(x,x^{\prime},t)=O(t)$ (and not $O(t^{2})$). A similar calculation
shows that we will also always have $\overline{S}(x,x^{\prime},t)-S_{1}%
(x,x^{\prime},t)=O(t)$. Denoting by $K_{1}(x,x^{\prime},t)$ and $K_{2}%
(x,x^{\prime},t)$ the approximate propagators obtained from the midpoint rules
(\ref{s1}) and (\ref{s2}), respectively, one checks without difficulty that we
will have%
\begin{align*}
\overline{K}(x,x^{\prime},t)-K_{1}(x,x^{\prime},t)  &  =O(t)\\
\overline{K}(x,x^{\prime},t)-K_{2}(x,x^{\prime},t)  &  =O(t)
\end{align*}
where $\overline{K}(x,x^{\prime},t)$ is the Kerner--Sutcliffe propagator
(\ref{KSpropagator}) (in these relations we can of course replace
$\overline{K}(x,x^{\prime},t)$ with the van Vleck propagator $\widetilde
{K}(x,x^{\prime},t)$ since both differ by a quantity $O(t^{2})$ in view of
Theorem \ref{Thm1bis}.

\section{The Case of Arbitrary Hamiltonians\label{sec2}}

\subsection{The main result}

We now consider the following very general situation: we assume that we are in
the presence of a quantum system represented by a state $|\psi\rangle$ whose
evolution is governed by a strongly continuous one-parameter group $(U_{t})$
of unitary operators acting on $L^{2}(\mathbb{R}^{n})$; the operator $U_{t}$
takes an initial wavefunction $\psi_{0}$ to $\psi=U_{t}\psi_{0}$. It follows
from Schwartz's kernel theorem \cite{69} that there exists a function
$K=K(x,x^{\prime};t)$ such that\footnote{This equality is sometimes
postulated; it is in fact a \textit{mathematical} fact which is true in quite
general situations.}
\begin{equation}
\psi(x,t)=\int K(x,x^{\prime};t)\psi_{0}(x^{\prime})d^{n}x^{\prime}
\label{psik}%
\end{equation}
and from Stone's \cite{Stone} theorem one strongly continuous one-parameter
groups of unitary operators that there exists a self-adjoint (generally
unbounded) operator $\widehat{H}$ on $L^{2}(\mathbb{R}^{n})$ such that
\begin{equation}
\psi(x,t)=e^{-\frac{i}{\hbar}\widehat{H}t}\psi_{0}(x); \label{Stone}%
\end{equation}
equivalently $\psi(x,t)$ satisfies the abstract Schr\"{o}dinger equation
(Jauch \cite{Jauch})%
\begin{equation}
i\hbar\frac{\partial\psi}{\partial t}(x,t)=\widehat{H}\psi(x,t).
\label{abschr}%
\end{equation}

We now make the following crucial assumption, which extrapolates to the
general case what we have done for Hamiltonians of the type classical type
\textquotedblleft kinetic energy plus potential\textquotedblright: the quantum
dynamics is again given by the Kerner--Sutcliffe propagator
(\ref{KSpropagator}) for small times $t$,\textit{ i.e.}%
\begin{equation}
K(x,x^{\prime},t)=\overline{K}(x,x^{\prime},t)+O(t^{2}) \label{approx}%
\end{equation}
the approximate propagator being given by
\begin{equation}
\overline{K}(x,x^{\prime},t)=\left(  \tfrac{1}{2\pi\hbar}\right)  ^{n}\int
e^{\tfrac{i}{\hbar}(p(x-x^{\prime})-\overline{H}(x,x^{\prime})t)}d^{n}p
\label{key}%
\end{equation}
where $\overline{H}$ is this time the averaged Hamiltonian function
\begin{equation}
\overline{H}(x,x^{\prime},p)=\int_{0}^{1}H(\tau x+(1-\tau)x^{\prime},p)d\tau.
\label{bjsymb}%
\end{equation}
Obviously, when $H=H_{\mathrm{free}}+V$ the function $\overline{H}$ reduces to
the function $H_{\mathrm{free}}+\overline{V}$ considered in Section \ref{sec1}.

This assumption can be motivated as follows (see de Gosson \cite{SPRINGER},
Proposition 15, \S 4.4). Let%
\[
S(x,x^{\prime},t)=\int_{\gamma}pdx-Hdt
\]
be Hamilton's two-point function calculated along the phase space path leading
from an initial point $(x^{\prime},p^{\prime},0)$ to a final point $(x,p,t)$
(the existence of such a function for small $t$ is guaranteed by
Hamilton--Jacobi theory; see \textit{e.g.} Arnol'd \cite{Arnold} or Goldstein
\cite{Gold}). That function satisfies the Hamilton--Jacobi equation%
\[
\frac{\partial S}{\partial t}+H(x,\nabla_{x}S)=0.
\]
One then shows that the function
\[
\overline{S}(x,x^{\prime},t)=p(x-x^{\prime})-\overline{H}(x,x^{\prime},p)t
\]
where $p$ is the momentum at time $t$ is an approximation to $S(x,x^{\prime
},t)$, in fact%
\[
\overline{S}(x,x^{\prime},t)-S(x,x^{\prime},t)=O(t^{2}).
\]
Here is an example: choose $H=\frac{1}{2}p^{2}x^{2}$ (we are assuming here
$n=1$); then%
\[
S(x,x^{\prime},t)=\frac{(\ln(x/x^{\prime}))^{2}}{2t}.
\]
Using the formula%
\[
\overline{H}(x,x^{\prime},p)=\frac{1}{6}p^{2}(x^{2}+xx^{\prime}+x^{\prime2})
\]
one shows after some calculations involving the Hamiltonian equations for $H$
that%
\[
\overline{S}(x,x^{\prime},t)=\frac{(\ln(x/x^{\prime}))^{2}}{2t}+O(t^{2})
\]
(see \cite{SPRINGER}, Chapter 4, Examples 10 and 16 for detailed calculations).

We are now going to show that the operator $\widehat{H}$ can be explicitly and
uniquely determined from the knowledge of $\overline{K}(x,x^{\prime},t)$.

\begin{theorem}
\label{Thm3}If we assume that the short-time propagator is given by formula
(\ref{key}) then the operator $\widehat{H}$ appearing in the abstract
Schr\"{o}dinger equation (\ref{abschr}) is given by%
\begin{equation}
\widehat{H}\psi(x)=\left(  \tfrac{1}{2\pi\hbar}\right)  ^{n}\int e^{\frac
{i}{\hbar}p(x-x^{\prime})}\overline{H}(x,x^{\prime},p)\psi(x^{\prime}%
)d^{n}pd^{n}x^{\prime}. \label{bjop}%
\end{equation}

\end{theorem}

\begin{proof}
Differentiating both sides of the equality (\ref{psik}) with respect to time
we get
\[
i\hbar\frac{\partial\psi}{\partial t}(x,t)=i\hbar\int\frac{\partial
K}{\partial t}(x,x^{\prime},t)\psi_{0}(x^{\prime})d^{n}x^{\prime};
\]
since $K$ itself satisfies the Schr\"{o}dinger equation (\ref{abschr}) we thus
have
\[
\widehat{H}\psi(x,t)=i\hbar\int\frac{\partial K}{\partial t}(x,x^{\prime
},t)\psi_{0}(x^{\prime})d^{n}x^{\prime}.
\]
It follows, using the assumptions (\ref{approx}) and (\ref{key}), that%
\[
\widehat{H}\psi(x,t)=i\hbar\int\frac{\partial\overline{K}}{\partial
t}(x,x^{\prime},t)\psi_{0}(x^{\prime})d^{n}x^{\prime}+O(t)
\]
and hence, letting $t\rightarrow0$,%
\begin{equation}
\widehat{H}\psi_{0}(x)=i\hbar\int\frac{\partial\overline{K}}{\partial
t}(x,x^{\prime},0)\psi_{0}(x^{\prime})d^{n}x^{\prime}. \label{clef}%
\end{equation}
Introducing the notation%
\[
\overline{S}(x,x^{\prime},t)=p(x-x^{\prime})-\overline{H}(x,x^{\prime},p)t
\]
we have
\begin{align*}
\frac{\partial\overline{K}}{\partial t}(x,x^{\prime},t)  &  =\left(  \tfrac
{1}{2\pi\hbar}\right)  ^{n}\tfrac{i}{\hbar}\int e^{\frac{i}{\hbar}\overline
{S}(x,x^{\prime},t)}\frac{\partial\overline{S}}{\partial t}(x,x^{\prime
},t)d^{n}p^{\prime}\\
&  =\left(  \tfrac{1}{2\pi\hbar}\right)  ^{n}\tfrac{1}{i\hbar}\int e^{\frac
{i}{\hbar}\overline{S}(x,x^{\prime},t)}\overline{H}(x,x^{\prime},p^{\prime
})d^{n}p^{\prime}.
\end{align*}
Taking the limit $t\rightarrow0$ and multiplying both sides of this equality
by $i\hbar$ we finally get%
\[
\widehat{H}\psi_{0}(x)=\left(  \tfrac{1}{2\pi\hbar}\right)  ^{n}\int
e^{\frac{i}{\hbar}p(x-x^{\prime})}\overline{H}(x,x^{\prime},p^{\prime
},t^{\prime})\psi_{0}(x^{\prime})d^{n}p^{\prime}d^{n}x^{\prime}%
\]
which proves (\ref{bjop}).
\end{proof}

We will call the operator $\widehat{H}$ defined by (\ref{bjop}) the
\emph{Born--Jordan quantization} of the Hamiltonian function $H$. That this
terminology is justified is motivated below.

\subsection{The case of monomials}

Let us show that (\ref{bjop}) reduces to the usual Born--Jordan quantization
rule (\ref{bj1}) when $H=x^{m}p^{\ell}$ (we are thus assuming dimension
$n=1$). We have here
\[
H(\tau x+(1-\tau)x^{\prime},p)=(\tau x+(1-\tau)x^{\prime})^{m}p^{\ell}%
\]
hence, using the binomial formula,%
\begin{equation}
H(\tau x+(1-\tau)x^{\prime},p)=\sum_{k=0}^{m}\binom{m}{k}\tau^{k}%
(1-\tau)^{m-k}x^{k}p^{\ell}x^{\prime m-k}. \label{amn}%
\end{equation}
Integrating from $0$ to $1$ in $\tau$ and noting that
\[
\int_{0}^{1}\tau^{k}(1-\tau)^{m-k}d\tau=\frac{k!(m-k)!}{(m+1)!}%
\]
we get%
\[
\overline{H}(x,x^{\prime},p)=\frac{1}{m+1}\sum_{k=0}^{m}x^{k}p^{\ell}x^{\prime
m-k}%
\]
and hence, using the definition (\ref{bjop}) of $\widehat{H}$,
\begin{align*}
\widehat{H}\psi(x)  &  =\frac{1}{2\pi\hbar(m+1)}\sum_{k=0}^{m}\int_{-\infty
}^{\infty}e^{\frac{i}{\hbar}p(x-x^{\prime})}x^{k}p^{\ell}x^{\prime m-k}%
\psi(x^{\prime})dpdx^{\prime}\\
&  =\frac{x^{k}}{2\pi\hbar(m+1)}\sum_{k=0}^{m}\int_{-\infty}^{\infty}\left(
\int_{-\infty}^{\infty}e^{\frac{i}{\hbar}p(x-x^{\prime})}p^{\ell}dp\right)
x^{\prime m-k}\psi(x^{\prime})dx^{\prime}.
\end{align*}
In view of the Fourier inversion formula (\ref{fou2}) we have%
\begin{equation}
\frac{1}{2\pi\hbar}\int_{-\infty}^{\infty}e^{\frac{i}{\hbar}p(x-x^{\prime}%
)}p^{\ell}dp=(-i\hbar)^{\ell}\delta^{(\ell)}(x-x^{\prime})
\end{equation}
so that we finally get%
\[
\widehat{H}\psi(x)=\frac{1}{m+1}\sum_{k=0}^{m}x^{k}(-i\hbar)^{\ell}%
\frac{\partial^{\ell}}{\partial x^{\ell}}(x^{m-k}\psi),
\]
which is equivalent to (\ref{bj1}) since $\widehat{p}^{\ell}=(-i\hbar)^{\ell
}\partial^{\ell}/\partial x^{\ell}$.

\subsection{Physical Hamiltonians\label{sec3}}

Let us now show that the Born--Jordan quantization of a physical Hamiltonian
of the type%
\begin{equation}
H=\sum_{j=1}^{n}\frac{1}{2m_{j}}(p_{j}-A_{j}(x))^{2}+V(x) \label{hyp}%
\end{equation}
coincide with the usual operator
\begin{equation}
\widehat{H}=\sum_{j=1}^{n}\frac{1}{2m_{j}}\left(  -i\hbar\frac{\partial
}{\partial x_{j}}-A_{j}(x)\right)  ^{2}+V(x) \label{claim}%
\end{equation}
obtained by Weyl quantization (the functions $A_{j}$ and $V$ are assumed to be
$C^{1}$). Since the quantizations of $p_{j}^{2}$, $A_{j}(x)$ and $V(x)$ are
the same in all quantization schemes (they are respectively $-\hbar
^{2}\partial^{2}/\partial x_{j}^{2}$ and multiplication by $A_{j}(x)$ and
$V(x)$), we only need to bother about the cross-products $p_{j}A(x)$. We claim that%

\begin{equation}
\widehat{p_{j}A}\psi=-\frac{i\hbar}{2}\left[  \frac{\partial}{\partial x_{j}%
}(A\mathcal{\psi)}+A\frac{\partial\psi}{\partial x_{j}}\right]  , \label{pwbj}%
\end{equation}
from which (\ref{claim}) immediately follows. Let us prove (\ref{pwbj}); it is
sufficient to do this in the case $n=1$. Denoting by $\overline{pA}$ the
Born--Jordan quantization of the function $pA$ we have
\[
\overline{pA}(x,x^{\prime},p)=p\int_{0}^{1}A(\tau x+(1-\tau)x^{\prime}%
)d\tau=p\overline{A}(x,x^{\prime})
\]
and hence
\begin{align*}
\widehat{pA}\psi(x)  &  =\frac{1}{2\pi\hbar}\int e^{\frac{i}{\hbar
}p(x-x^{\prime})}p\overline{A}(x,x^{\prime})\psi(x^{\prime})dx^{\prime}dp\\
&  =\int_{-\infty}^{\infty}\left(  \frac{1}{2\pi\hbar}\int_{-\infty}^{\infty
}e^{\frac{i}{\hbar}p(x-x^{\prime})}pdp\right)  \overline{A}(x,x^{\prime}%
)\psi(x^{\prime})dx^{\prime}.
\end{align*}
In view of (\ref{fou2}) the expression between the square brackets is
$-i\hbar\delta^{\prime}(x-x^{\prime})$ so that%
\begin{align*}
\widehat{pA}\psi(x)  &  =-i\hbar\int_{-\infty}^{\infty}\delta^{\prime
}(x-x^{\prime})\overline{A}(x,x^{\prime})\psi(x^{\prime})dx^{\prime}\\
&  =-i\hbar\int_{-\infty}^{\infty}\delta(x-x^{\prime})\frac{\partial}{\partial
x^{\prime}}(\overline{A}(x,x^{\prime})\psi(x^{\prime}))dx^{\prime}\\
&  =-i\hbar\left(  \frac{\partial\overline{A}}{\partial x^{\prime}}%
(x,x)\psi(x))+\overline{A}(x,x)\frac{\partial\psi}{\partial x^{\prime}%
}(x))\right)
\end{align*}
Now, by definition of $\overline{A}(x,x^{\prime})$ we have $\overline
{A}(x,x)=A(x)$ and%
\[
\frac{\partial\overline{A}}{\partial x^{\prime}}(x,x)=\int_{0}^{1}%
(1-\tau)\frac{\partial A}{\partial x}(x)d\tau=\frac{1}{2}\frac{\partial
A}{\partial x}(x)
\]
and hence%
\[
\widehat{pA}\psi=-\frac{i\hbar}{2}\frac{\partial A}{\partial x}\psi-i\hbar
A\frac{\partial\psi}{\partial x}%
\]
which is the same thing as (\ref{pwbj}).

\section{Discussion}

Both Kerner and Sutcliffe, and Cohen relied on path integral arguments which
were doomed to fail because of the multiple possible choices of histories in
path integration. However, it follows from our rigorous constructions that
Kerner and Sutcliffe's insight was right, even though it was not
mathematically justified. While there is, as pointed out by Cohen
\cite{Cohenwrong}, a great latitude in choosing the short-time propagator,
thus leading to different quantizations, our argument did not make use of any
path-integral argument; what we did was to propose a short-time propagator
which is \emph{exact} up to order $O(t^{2})$ (as opposed to those obtained by
using midpoint rules), and to show that if one use this propagator, then one
must quantize Hamiltonian functions (and in particular monomials) following
the prescription proposed by Born and Jordan in the case of monomials.

\begin{acknowledgement}
This work has been financed by the Austrian Science Fund (FWF). Grant number:
P 27773--N25. The author thanks three Reviewers for very useful comments, and
for having pointed out several typos.
\end{acknowledgement}

\end{document}